\documentclass[11pt]{article}
\usepackage{mathrsfs}
\usepackage{amssymb}
\usepackage{amsfonts}
\usepackage{color}
\usepackage{amsthm}
\usepackage{amsmath}
\usepackage[top=35truemm, bottom=30truemm, left=30truemm, right=30truemm]{geometry}
\usepackage[pdftex]{graphicx}
\usepackage{latexsym}
\newtheorem{dfn}{Definition}
\newtheorem{thm}{Theorem}[section]
\newtheorem{lem}[thm]{Lemma}

\title{The second law-type work relation \\in non-equilibrium steady states\\ in one-dimensional quantum lattice systems}
\author{Kazuki Yamaga\thanks{Department of Nuclear Engineering, Kyoto University, yamaga.kazuki.62a@st.kyoto-u.ac.jp}}
\date{}

\begin{document}
\maketitle
\begin{abstract}
We consider the Non-Equilibrium Steady State induced by two infinite quantum thermal reservoirs at different temperatures $\beta_L^{-1},\beta_R^{-1}$ and derive an inequality giving the upper bound of the work extracted by cyclic operations. This upper bound tends to $0$ in the equilibrium limit $\beta_L\to\beta_R$ and the inequality reproduces the second law of thermodynamics that one cannot extract any work from equilibrium states by local cyclic operations. In addition, we consider global cyclic operations and obtain an upper bound of the work density in one-dimensional quantum lattice systems, which depends on the model and $\beta_L,\beta_R$. This bound is independent of the operations and also tends to 0 in the equilibrium limit.
\end{abstract}
\section{Introduction}
Quantum statistical mechanics bridges the microscopic quantum dynamics and the macroscopic world. In equilibrium systems there is a complete description by KMS (Gibbs) states. However, we know little about non-equilibrium systems. Most of the studies of non-equilibrium systems focuses on the physically important classes such as linear response regime\cite{K} and non-equilibrium steady states (NESS). In this paper we will deal with the NESS induced by two infinite quantum thermal reservoirs at different temperatures $\beta_L^{-1},\beta_R^{-1}$. We aim to derive the universal properties of NESS with the stand point of work and operations. In equilibrium systems there is a universal law, the second law of thermodynamics that one cannot extract any work by cyclic operations from the system in equilibrium. We explore the second law-type work relation in NESS that reduces to the second law of thermodynamics in the limit that the difference of the temperatures of the two reservoirs tends to $0$.

Using the scattering approach of NESS we derived an inequality which can be regarded as the extension of the second law of thermodynamics. In this inequality the work extracted by a cyclic operation is bounded above by a non-trivial positive constant depending on the operation. In the equilibrium limit $\beta_L\to\beta_R$ this bound becomes $0$ and the inequality reproduces the second law of thermodynamics. Furthermore we estimate the work density (the ratio of the work to the volume of the support of the operation) in one-dimensional quantum lattice systems and derive an upper bound independent of the operations. This bound is also $0$ in the equilibrium limit $\beta_L\to\beta_R$, although the inequality is not tight in the sense that there are no operations which achieve this bound.

This paper is organized as follows. In section 2, we introduce the situation which we will consider and review the scattering approach of NESS. Section 3 is devoted to our main result, the second law-type work relation in NESS. Under some assumptions on the dynamics we first derive an inequality on the upper bound of the work extracted from NESS by cyclic operations. Then the work density is estimated in one-dimensional quantum lattice systems. Finally in section 4, we discuss free Fermi gas on the one-dimensional lattice $\mathbb{Z}$, it is an example that satisfies all the assumptions in section 3.

\section{Non-Equilibrium Steady States}

A general quantum system including an infinitely extended system is described by a (unital) C*-algebra $\mathcal{A}$ and a one-parameter group of *-automorphisms $\{\alpha_t\mid t\in\mathbb{R}\}$ on $\mathcal{A}$, which represent the set of observables of the system and the dynamics respectively. In the present paper, we assume that the dynamics $\alpha_t$ is strongly continuous, i.e. $\displaystyle\lim_{t\to0}\|\alpha_t(A)-A\|=0,\ A\in\mathcal{A}$. Such a pair $(\mathcal{A},\alpha)$ is called a C*-dynamical system. State is a normalized positive linear functional $\omega\colon\mathcal{A}\to\mathbb{C}$ giving the expectation values of observables;
\begin{itemize}
\item $\omega(\lambda A+B)=\lambda\omega(A)+\omega(B),\ \lambda\in\mathbb{C},A,B\in\mathcal{A}$
\item $\omega(A^{*}A)\ge0,\ A\in\mathcal{A}$
\item $\omega(I)=1$.
\end{itemize}
Denote the set of  all states of $\mathcal{A}$ by $\mathfrak{S}(\mathcal{A})$. $\mathfrak{S}(\mathcal{A})$ is compact in the weak* topology. In equilibrium quantum statistical mechanics, thermal equilibrium state is described by the following KMS state.
\begin{dfn}[KMS state]
Let $\mathcal{A}$ be a C*-algebra and $\{\alpha_t\mid t\in\mathbb{R}\}$ a dynamics on $\mathcal{A}$. For $\beta>0$ the state $\omega$ satisfying the following conditions is called a $(\beta,\alpha)$-KMS state;\\
For any $A,B\in\mathcal{A}$, there exists a function $F_{AB}(z)$ analytic on $\mathcal{D}_{\beta}=\{z\in\mathbb{C}\mid0<\mathrm{Im}z<\beta\}$ and bounded continuous on $\overline{\mathcal{D}_{\beta}}$ (the closure of $\mathcal{D}_\beta$) and satisfying the boundary conditions
\[ F_{AB}(t)=\omega(A\alpha_t(B)) \]
\[ F_{AB}(t+i\beta)=\omega(\alpha_t(B)A),\ t\in\mathbb{R} .\]
\end{dfn}

Let us now introduce Non-Equilibrium Steady State (NESS). There are several approaches to the study of NESS such as using quantum dynamical semigroup\cite{LS,JP} and Hamiltonian dynamics of infinite systems including reservoirs. In this paper we consider NESS induced by infinitely extended reservoirs introduced by Ruelle\cite{R}. Here we recall it.

Let $(\mathcal{A},\alpha)$ be a C*-dynamical system and $\omega$  a $\alpha$-invariant state, i.e. $\omega\circ\alpha_t=\omega,\ t\in\mathbb{R}$. $\alpha^V_t$ is the dynamics perturbed by $V=V^{*}\in\mathcal{A}$,
\[ \alpha_t^V(A)=\alpha_t(A)+\sum_{n=1}^\infty i^n\int^t_0dt_1\int^{t_1}_0dt_2\cdots\int^{t_{n-1}}_0dt_n[\alpha_{t_n}(V),[\cdots,[\alpha_{t_1}(V),\alpha_t(A)]\cdots] .\]
If there is an increasing sequence $T_n\nearrow\infty$ such that  
\[ \omega_+(A)=\lim_{n\to\infty}\frac{1}{T_n}\int_0^{T_n}\omega\circ\alpha_t^V(A)dt,\ A\in\mathcal{A}, \]
then we say that $\omega_+$ is a NESS. Denote $\Sigma_V^+(\omega)$ the set of NESS starting from the initial state $\omega$ . Since $\mathfrak{S}(\mathcal{A})$ is compact $\Sigma_V^+(\omega)$ is not empty set. It is easily checked that NESS $\omega_+\in\Sigma^+_V(\omega)$ is  $\alpha^V$-invariant. NESS is the finally realized state developed by the dynamics $\alpha^V_t$ starting from the initial state $\omega$. 

It is a difficult problem to show the convergence to NESS, $\displaystyle\lim_{t\to\infty}\omega\circ\alpha_t^V(A)=\omega_+(A),\ A\in\mathcal{A}$ in concrete models. The models that have been exactly proved are free Fermi gas (we will deal with in section 4) and XY model \cite{AH,T,JP1,AP,AJPP2}. There are two approaches to deal with the convergence to NESS, the scattering approach and the spectral approach\cite{JP2}. In this paper, we adopt the scattering approach. In this approach the key assumption is the existence of the limit s-$\displaystyle\lim_{t\to\infty}\alpha_{-t}\circ\alpha_t^V$ (s-$\displaystyle\lim_{t\to\infty}$ is the strong limit). If the limit s-$\displaystyle\lim_{t\to\infty}\alpha_{-t}\circ\alpha_t^V\equiv\gamma$ exists, due to the $\alpha$-invariance of $\omega$ one obtains the convergence to NESS
\[ \omega\circ\alpha^V_t(A)=\omega\circ\alpha_{-t}\circ\alpha^V_t(A)\to\omega\circ\gamma(A)\ (t\to\infty). \]
Simple calculation shows 
\begin{itemize}
\item $\alpha_t\circ\gamma=\gamma\circ\alpha^V_t$
\item $\gamma$ is isometry 
\item $\gamma:\mathcal{A}\to\mathcal{A}$ is a *-morphism.
\end{itemize}
Note that in general $\gamma$ is not surjective. By the relation $\alpha_t\circ\gamma=\gamma\circ\alpha^V_t$, one obtains $A\in\mathrm{Dom}(\delta_V)\iff\gamma(A)\in\mathrm{Dom}(\delta)$ and $\delta\circ\gamma(A)=\gamma\circ\delta_V(A)$. $\delta$ and $\delta_V=\delta+i[V,\cdot]$ are the generators of $\alpha_t$ and $\alpha_t^V$ respectively and $\mathrm{Dom}(\delta)=\mathrm{Dom}(\delta_V)$ is its domain.

For example, if the system has the $L^1$-asymptotic abelianness, $\alpha_{-t}\circ\alpha_t^V$ converges\cite{AJPP}. 

We close this section by introducing the situation discussed in this paper. Suppose that two thermal reservoirs at temperatures $\beta_L^{-1},\beta_R^{-1}$ contact through a small system (a small system between the reservoirs is not always necessary), see figure1. This situation can be realized in mesoscopic systems. 

\begin{figure}[htbp]
\begin{center}
\includegraphics[clip,width=8.0cm]{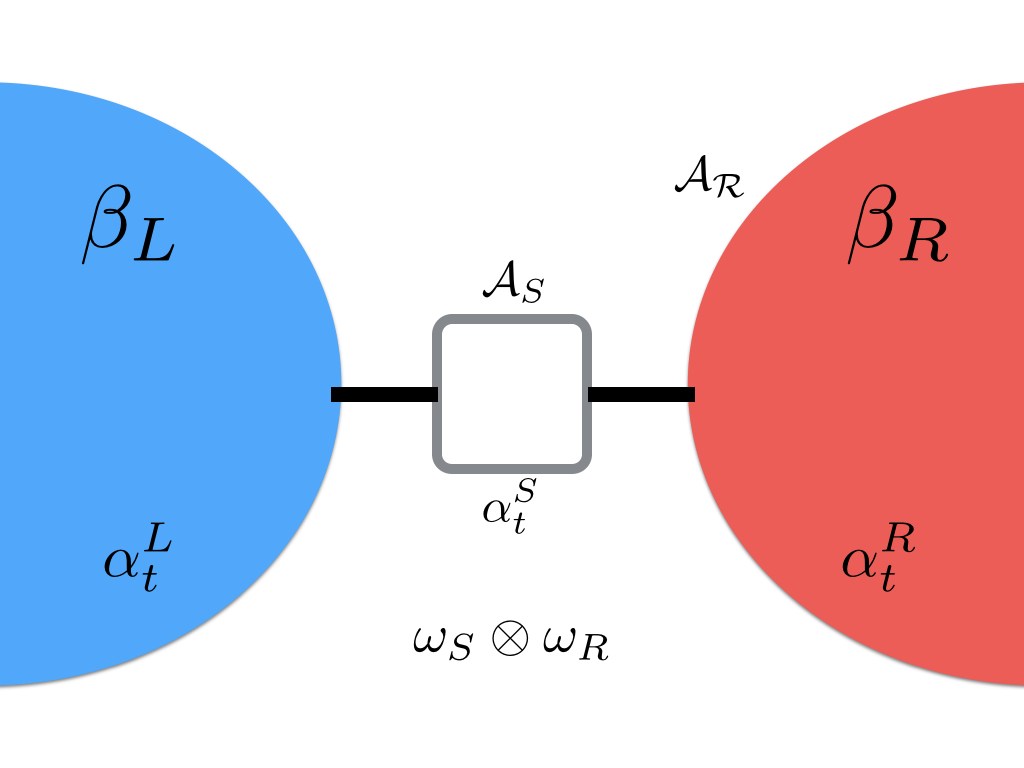}
\end{center}
\caption{a small system with two thermal reservoirs}
\end{figure}

The small system between the reservoirs is described by a finite dimensional C*-algebra $\mathcal{A}_\mathcal{S}=M(n,\mathbb{C})$ and the C*-algebra of observables of the reservoirs is denoted by $\mathcal{A}_\mathcal{R}$. The whole system is the composite system $\mathcal{A}=\mathcal{A}_\mathcal{S}\otimes\mathcal{A}_{\mathcal{R}}$. The dynamics of the small system is given by 
\[ \alpha^\mathcal{S}_t(A)=e^{itH_\mathcal{S}}Ae^{-itH_\mathcal{S}},\ t\in\mathbb{R} \]
for a self-adjoint operator $H_\mathcal{S}\in\mathcal{A}_\mathcal{S}$ (Hamiltonian). The dynamics of the reservoirs $\alpha_t^\mathcal{R}$ is given by the product of the commuting C*-dynamics $\alpha_t^L,\alpha_t^R$;
\[ \alpha^{\mathcal{R}}_t=\alpha^L_t\circ\alpha^R_t .\]
$\alpha^L_t$ is the dynamics of the left reservoir and $\alpha^R_t$ is that of the right reservoir. The dynamics of the whole system before the interaction $V=V^{*}\in\mathcal{A}$ is switched on is $\alpha_t=\alpha^\mathcal{S}_t\otimes\alpha^\mathcal{R}_t$. The system is initially prepared in the state $\omega=\omega_\mathcal{S}\otimes\omega_\mathcal{R}$, where $\omega_\mathcal{S}(A)=\frac{1}{n}\mathrm{Tr}A$  ($A\in\mathcal{A}_\mathcal{S}$) is the chaotic state (since we are interested in the limit state, NESS, the choice of the state in the small system at the initial time does not matter) and $\omega_\mathcal{R}$ is an (1,$\tilde{\alpha}$)-KMS state ($\tilde{\alpha}_t=\alpha^L_{\beta_L t}\circ\alpha^R_{\beta_R t}$). That is, the left and right reservoirs are in equilibrium at temperatures $\beta_L^{-1},\beta_R^{-1}$ respectively. Here after we assume that the left reservoir is colder than the right one, $\beta_L>\beta_R$. Obviously this initial state is $\alpha$-invariant. 

Examples of this situation contain the following cases.
\begin{description}
\item[(1)]Two reservoirs are described by C*-dynamical systems $(\mathcal{A}_L,\alpha^L),\ (\mathcal{A}_R,\alpha^R)$ respectively and consider the composite system 
\[ \mathcal{A}_\mathcal{R}=\mathcal{A}_L\otimes\mathcal{A}_R \]
\[ \alpha^\mathcal{R}_t=\alpha^L_t\otimes\alpha^R_t=(\alpha^L_t\otimes id_R)\circ(id_L\otimes\alpha^R_t), \]
where $id_{\#}$ is the identity operator on $\mathcal{A}_\#$ ($\#=L,R$).
Quantum spin systems are included in this situation.
\item[(2)]Free Fermi reservoir. 
\[ \mathcal{A}_\mathcal{R}=\mathcal{A}^{CAR}(\mathcal{H}_L\oplus\mathcal{H}_R) \]
\[ \alpha^L_t(a^{\#}(f))=a^{\#}(e^{ith_L}f )\]
\[ \alpha^R_t(a^{\#}(f))=a^{\#}(e^{ith_R}f ) ,\ f\in\mathcal{H}_L\oplus\mathcal{H}_R, \]
where $a^{\#}$ is $a^{*}$ or $a$. For a 1-particle Hilbert space $\mathcal{H}$, the CAR algebra $\mathcal{A}^{CAR}(\mathcal{H})$ is the C*-algebra generated by creation and annihilation operators $a^{*}(f),a(f),\ f\in\mathcal{H}$ satisfying the canonical anti-commutation relations
\[ \{a(f),a^{*}(g)\}=\langle f,g\rangle I,\ \{a(f),a(g)\}=0,\]
where $\{A,B\}=AB+BA$ and $\langle\cdot,\cdot\rangle$ is the inner product of $\mathcal{H}$. $h_L,h_R$ are self-adjoint operators on $\mathcal{H}_L,\mathcal{H}_R$ respectively (1-particle Hamiltonian).
\end{description}

\section{Extension of The Second Law of Thermodynamics}
This section is the main part of the present paper. We assume that NESS can be constructed by the scattering approach. And we consider the work extracted from the NESS by cyclic operations.

For a quantum system $(\mathcal{A},\alpha)$, an operation done by the outside world is given as a time-dependent perturbation $V(t)=V(t)^{*}\in\mathcal{A},\ t\in[0,T]$. An operation $V(t)$ is called cyclic if $V(0)=V(T)$. The dynamics $\theta_t\colon\mathcal{A}\to\mathcal{A}$ induced by the operation $V(t)$ is determined by the equations
\begin{itemize}
\item $\frac{d}{dt}\theta_t(A)=\theta_t(\delta(A)+i[V(t),A])$
\item $\theta_0(A)=A,\ A\in\mathcal{A}$.
\end{itemize}
$\delta$ is the generator of $\alpha_t$. $\theta_t$ can be written as
\[ \theta_t(A)=U_t^{*}\alpha_t(A)U_t, \]
where $U_t$ is the unitary element in $\mathcal{A}$ determined by the equations
\begin{itemize}
\item $\frac{d}{dt}U_t=-i\alpha_t(V(t))U_t$
\item $U_0=I$.
\end{itemize}
The work extracted by the cyclic operation $V(t)$ is defined as  $i\omega(U_T^{*}\delta(U_T))$. This type of work is discussed in \cite{L,PW}. In the case of a finite system, $\omega(A)=\mathrm{Tr}\rho A,\alpha_t(A)=e^{itH}Ae^{-itH}$, this is $\mathrm{Tr}\rho H-\mathrm{Tr}\rho\theta_T(H)$ , the difference of energy in the initial and in the finial state.

It is a well known fact that any work cannot be extracted from thermal equilibrium states (KMS states) by cyclic operations (the second law of thermodynamics):
\begin{lem}
Let $\omega$ be a $(\beta,\alpha)$-KMS state, then for a unitary element $U\in\mathrm{Dom}(\delta)$
\[  i\omega(U^{*}\delta(U))=-\frac{1}{\beta}S(\omega\circ\mathrm{Ad}U\|\omega)\le0 \]
holds, where $\mathrm{Ad}U(A)=U^{*}AU$ and $S(\omega\circ\mathrm{Ad}U\|\omega)$ is Araki's relative entropy \cite{A}.
\end{lem}

It should be noted that if $\omega$ is a state describing pure phase but not a KMS state or a ground state for the dynamics $\alpha_t$, there is a cyclic operation with which one can extract positive work. This follows from the result of \cite{PW}.

\subsection{General Theory}
Here we impose two assumptions.
\begin{description}
\item[(A1)]The limit s-$\displaystyle\lim_{t\to\infty}\alpha_{-t}\circ\alpha^V_t=\gamma$ exists.
\item[(A2)]$\mathrm{Ran}\gamma=I\otimes\mathcal{A}_\mathcal{R}$.
\end{description}
As explained in the previous section, the system converges to the NESS $\omega_+=\omega\circ\gamma$. This NESS $\omega_+$ is $(1,\sigma)$-KMS state, where $\sigma_t=\gamma^{-1}\circ\tilde{\alpha}_t\circ\gamma$. By Lemma 3.1 for a unitary element $U\in\mathrm{Dom}(\delta_V)$,
\[ -i\omega_+(U^{*}\gamma^{-1}\circ(\beta_L\delta_L+\beta_R\delta_R)\circ\gamma(U))=S(\omega_+\circ\mathrm{Ad}U\|\omega_+) \]
holds. $\gamma^{-1}$ makes sense because $(\beta_L\delta_L+\beta_R\delta_R)\circ\gamma(U)\in I\otimes\mathcal{A}_{\mathcal{R}}$. By the assumption (A2) and the equation $\alpha_t\circ\gamma=\gamma\circ\alpha_t^V$, we have $(\delta_L+\delta_R)\circ\gamma(U)=\gamma\circ\delta_V(U)$ for $U\in\mathrm{Dom}(\delta_V)$ and 
\begin{eqnarray*}
&& -i\omega_+(U^{*}\gamma^{-1}\circ(\beta_L\delta_L+\beta_R\delta_R)\circ\gamma(U)) \\
&=&-i\omega_+\left(U^{*}\gamma^{-1}\circ\left(\frac{\beta_L+\beta_R}{2}(\delta_L+\delta_R)+\frac{\beta_L-\beta_R}{2}(\delta_L-\delta_R)\right)\circ\gamma(U)\right)\\
&=&-i\left(\frac{\beta_L+\beta_R}{2}\right)\omega_+(U^{*}\delta_V(U))-i\left(\frac{\beta_L-\beta_R}{2}\right)\omega_+(U^{*}\gamma^{-1}\circ(\delta_L-\delta_R)\circ\gamma(U)).
\end{eqnarray*}

\begin{lem}
Under the above two assumptions $\mathrm{(A1), (A2)}$, for NESS $\omega_+=\omega\circ\gamma$ and a unitary element $U\in\mathrm{Dom}(\delta_V)$
\begin{eqnarray*}
i\omega_+(U^{*}\delta_V(U))+\frac{2}{\beta_L+\beta_R}S(\omega_+\circ\mathrm{Ad}U\|\omega_+)&=&-i\frac{\beta_L-\beta_R}{\beta_L+\beta_R}\omega_+(U^{*}\gamma^{-1}\circ(\delta_L-\delta_R)\circ\gamma(U))\\
&\le&\frac{\beta_L-\beta_R}{\beta_L+\beta_R}\|\gamma^{-1}\circ(\delta_L-\delta_R)\circ\gamma(U)\|.
\end{eqnarray*}
holds.
\end{lem}

Since the relative entropy is non-negative we obtain an upper bound of the work.
\begin{thm}
Under the assumptions $\mathrm{(A1), (A2)}$, we have
\begin{equation}
i\omega_+(U^{*}\delta_V(U))\le\frac{\beta_L-\beta_R}{\beta_L+\beta_R}\|\gamma^{-1}\circ(\delta_L-\delta_R)\circ\gamma(U)\|.
\end{equation}
\end{thm}
This inequality means that the work extracted by the cyclic operation $U$ from the NESS $\omega_+$ is bounded above by $\frac{\beta_L-\beta_R}{\beta_L+\beta_R}\|\gamma^{-1}\circ(\delta_L-\delta_R)\circ\gamma(U)\|$. This bound is small when the temperatures of the two reservoirs are close and the average temperature is high. Furthermore in the equilibrium limit $\beta_L\to\beta_R$ this bound tends to 0, this inequality reproduces the second law of thermodynamics (Lemma 3.1). Although in the second law of thermodynamics the bound 0 is independent of the cycle, the bound in the inequality (1) is not. Generally, for NESS such a cycle-independent upper bound cannot be expected. In fact as discussed later, one can extract exactly positive work density. This means that there is no finite cycle-independent bound for work. Then how about ``work density'' defined by the ratio of the work to the volume of the support of the operation? In the next section we consider one-dimensional quantum lattice systems and estimate the work density extracted by cyclic operations.

\subsection{Work Density in One-Dimensional Quantum Lattice Systems}
Here we consider one-dimensional quantum lattice systems, quantum spin systems or lattice fermion systems on $\mathbb{Z}$. First let us recall them.
\begin{itemize}
\item Quantum Spin Systems;\\
Let $\mathcal{H}$ be a finite dimensional Hilbert space describing a spin. For a finite region $X\in\mathcal{P}_f(\mathbb{Z})$ ($\mathcal{P}_f(\mathbb{Z})$ is the set of all finite subsets of $\mathbb{Z}$) the algebra of the observables on $X$ is $\mathcal{A}^S_X=\mathcal{B}(\mathcal{H}_X)$, where $\mathcal{H}_X=\mathcal{H}^{\otimes|X|}$ ($|X|$ is the number of elements of $X$). Define a C*-algebra $\mathcal{A}^S$ of quasi-local observables by a completion of the *-algebra of strictly local observables $\mathcal{A}^S_{loc}=\displaystyle\cup_{X\in\mathcal{P}_f(\mathbb{Z})}\mathcal{A}^S_X$,
\[ \mathcal{A}^S=\overline{\mathcal{A}_{loc}}^{\|\cdot\|}. \]
For infinite $\Lambda\subset\mathbb{Z}$, $\mathcal{A}_\Lambda^S$ is defined similarly, $\mathcal{A}_\Lambda^S=\overline{\cup_{X\in\mathcal{P}_f(\Lambda)}\mathcal{A}_X^S}^{\|\cdot\|}$.
\item Lattice Fermion Systems;\\
Denote $\mathcal{A}^{CAR}=\mathcal{A}^{CAR}(l^2(\mathbb{Z}))$, where $l^2(\mathbb{Z})$ is the Hilbert space of square summable functions on $\mathbb{Z}$. For $X\subset\mathbb{Z}$, define $\mathcal{A}_X^{CAR}$ as the C*-algebra generated by $a_n,a^{*}_n,\ n\in X$ (we denote $a^{\#}(e_n)$ by $a^{\#}_n$, where $\{e_n\}_{n\in\mathbb{Z}}$ is the standard basis of $l^2(\mathbb{Z})$; $e_n(m)=\delta_{nm}$). Define a *-automorphism $\Theta$ of $\mathcal{A}^{CAR}$ by 
\[ \Theta(a^{\#}_n)=-a^{\#}_n,\ n\in\mathbb{Z} \]
Obviously $\Theta^2=id$. The elements of 
\[ \mathcal{A}^{CAR}_+=\{A\in\mathcal{A}^{CAR}\mid \Theta(A)=A\} \]
\[ \mathcal{A}^{CAR}_-=\{A\in\mathcal{A}^{CAR}\mid\Theta(A)=-A \} \]
are called even, odd respectively. Any $A\in\mathcal{A}^{CAR}$ can be decomposed into even and odd parts
\[ A=A_++A_-,\]
$A_+\in\mathcal{A}^{CAR}_+,A_-\in\mathcal{A}^{CAR}_-$. So $\mathcal{A}^{CAR}=\mathcal{A}^{CAR}_++\mathcal{A}^{CAR}_-$.
\end{itemize}
Let $\mathcal{A}$ be $\mathcal{A}^S$ or $\mathcal{A}^{CAR}$ and $\mathcal{A}_X$ be $\mathcal{A}_X^S$ or $\mathcal{A}_X^{CAR}$ for $X\subset\mathbb{Z}$. $\mathcal{A}$ has the action of the translation $\mathbb{Z}$, $v_x\colon\mathcal{A}\to\mathcal{A},\ x\in\mathbb{Z}$. $\mathcal{A}$ has the unique tracial state $\tau$ ($\tau$ is a state satisfying $\tau(AB)=\tau(BA)$ for $A,B\in\mathcal{A}$) and the unique conditional expectation $E_\Lambda\colon\mathcal{A}\to\mathcal{A}$ for each $\Lambda\subset\mathbb{Z}$ characterized by the following properties
\begin{itemize}
\item $E_\Lambda(A)\in\mathcal{A}_\Lambda$
\item $\tau(E_\Lambda(A)B)=\tau(AB)$ for $A\in\mathcal{A},B\in\mathcal{A}_\Lambda$.
\end{itemize}
Conditional expectation $E_\Lambda$ is a positive unital linear map with norm 1 that satisfies $E_\Lambda^2=E_\Lambda$ and for $A\in\mathcal{A},B\in\mathcal{A}_\Lambda$
\begin{itemize}
\item $E_\Lambda(AB)=E_\Lambda(A)B$
\item $E_\Lambda(BA)=BE_\Lambda(A)$.
\end{itemize}
This conditional expectation satisfies $E_\Lambda\cdot E_{\Lambda'}=E_{\Lambda\cap\Lambda'}$.

The important notion of the lattice systems is the interaction. Interaction $\Phi$ is a function mapping a finite region $X\in\mathcal{P}_f(\mathbb{Z})$ to a self-adjoint element $\Phi(X)$ in $\mathcal{A}_X$. $\Phi$ is called translationally invariant if $v_x(\Phi(X))=\Phi(X+x)$ for $x\in\mathbb{Z}, X\in\mathcal{P}_f(\mathbb{Z})$ ($X+x=\{y+x\mid y\in X\}$). If there is $R>0$ such that $\mathrm{diam}(X)\ge R\Rightarrow\Phi(X)=0$, then we say that $\Phi$ is of finite range, where $\mathrm{diam}(X)$ is a diameter of $X$, $\mathrm{diam}(X)=\sup\{|x-y|\mid x,y\in X\}$. In this section we consider operations done by the outside world not necessarily local. Such operations are formulated as time dependent interactions. Given a time dependent interaction $\Phi_t,\ t\in[0,T]$, for each finite region $\Lambda\in\mathcal{P}_f(\mathbb{Z})$ we have a time dependent local Hamiltonian 
\[ H_\Lambda(\Phi_t)\equiv\sum_{X\subset\Lambda}\Phi_t(X) \in\mathcal{A}_\Lambda \]
and the dynamics $\theta_t^\Lambda$ following the discussion in the introduction of this section. This $\theta_t^\Lambda$ can be written as 
\[ \theta_t^\Lambda(A)=U_t^{\Lambda*}AU_t^\Lambda,\ A\in\mathcal{A} \]
for a unitary element $U_t^\Lambda\in\mathcal{A}_\Lambda$. Extend $\Phi_t$ for $t\neq[0,T]$ by $\Phi_t=\Phi_0,\ t\in(-\infty,0)$ and $\Phi_t=\Phi_T,\ t\in(T,\infty)$ and define 
\[ \|\Phi_\cdot\|\equiv\sup_{x,y\in\mathbb{Z}}\sum_{Z\ni x,y}(|x-y|+1)^2\sup_{t\in\mathbb{R}}\|\Phi_t(Z)\| .\]
If $\|\Phi_\cdot\|<\infty$ the thermodynamical limit of $\theta^\Lambda_t$ exists. Consider the sequence of finite subsets of the form $\Lambda_N=[-N,N]\cap\mathbb{Z},\ N\in\mathbb{N}$.

\begin{lem}
For an operation $\Phi_t$ such that $\|\Phi_\cdot\|<\infty$ (and even, i.e. $\Phi_t(X)\in\mathcal{A}_+^{CAR}$,  in the case of the fermion system), the limit
\[ \lim_{N\to\infty}\theta_t^{\Lambda_N}(A)=\theta_t(A),\ A\in\mathcal{A} \]
exists and $\theta_t$ is an isometric *-morphism of $\mathcal{A}$. If $\Phi_t$ does not depend on $t\in\mathbb{R}$ $(\Phi_t=\Phi)$, $\theta_t$ is a C*-dynamics.
\end{lem}

The proof of this lemma is similar to that of time independent interaction case\cite{NO}. 

Here we constructed the dynamics from the interaction. Conversely given a dynamics we can construct an interaction from it \cite{BR}. Let $\alpha_t$ be a C*-dynamics on $\mathcal{A}$ and $\mathcal{A}_{loc}\subset\mathrm{Dom}(\delta)$ ($\delta$ is the generator of $\alpha_t$). Then for each $\Lambda\in\mathcal{P}_f(\mathbb{Z})$ there is $\tilde{H}_\Lambda\in\mathcal{A}$ such that 
\[ \delta(A)=i[\tilde{H}_\Lambda,A],\ A\in\mathcal{A}_\Lambda .\]
$\tilde{H}_\Lambda$ has a freedom to add elements of $\mathcal{A}^S_{\Lambda^C},\mathcal{A}^{CAR}_{\Lambda^C}\cap\mathcal{A}_+^{CAR}$. Choose $\tilde{H}_\Lambda$ so that $E_{\Lambda^C}(\tilde{H}_\Lambda)=0$. Set $H_\Lambda=E_\Lambda(\tilde{H}_\Lambda)$ and define $\Phi$ inductively by 
\[ \Phi(\emptyset)=0 \]
\[ \Phi(X)=H_X-\sum_{Y\subsetneq X}\Phi(Y)\in\mathcal{A}_X. \]
Then $\Phi$ is an interaction and by definition $H_\Lambda=\displaystyle\sum_{X\subset\Lambda}\Phi(X)$.
Furthermore 
\begin{equation}
\sum_{X\cap\Lambda\neq\emptyset}\Phi(X) \left(=\lim_{\Lambda'\to\mathbb{Z}}\sum_{X\cap\Lambda\neq\emptyset,X\subset\Lambda'}\Phi(X) \right)
\end{equation}
converges and $\tilde{H}_\Lambda=\displaystyle\sum_{X\cap\Lambda\neq\emptyset}\Phi(X)$. Note that the interaction $\Phi$ constructed above satisfies that if $Y\subsetneq X$, then $E_Y(\Phi(X))=0$.

Now consider NESS of the lattice system. Divide the one-dimensional lattice $\mathbb{Z}$ into three parts $(-\infty,-M-1], [-M,M],[M+1,\infty),\ M\in\mathbb{N}$, the left reservoir, the small system and the right reservoir.

Suppose that the dynamics $\alpha_t^V$ is given by a translationally invariant finite range interaction $\Phi$ (in the case of the fermion system $\Phi$ is even). The interaction between the small system and the reservoirs is $V=\displaystyle\sum_{X\cap\Lambda_M\neq\emptyset,X\cap\Lambda_M^C\neq\emptyset}\Phi(X)\in\mathcal{A}$. Here we impose an assumption in addition to the assumptions (A1), (A2).
\begin{description}
\item[(A3)]For any $\eta\in\mathcal{N}_\omega$, weak*-$\displaystyle\lim_{t\to\infty}\eta\circ\alpha_t^V=\omega_+$.
\end{description}
$\mathcal{N}_\omega$ is the set of $\omega$-normal states. $\phi\in\mathfrak{S}(\mathcal{A})$ is said to be $\omega$-normal if there is a density operator $\rho$ on $\mathcal{H}_\omega$ such that $\phi(A)=\mathrm{Tr}\rho\pi_\omega(A)$, where $(\pi_\omega,\mathcal{H}_\omega)$ is the GNS representation associated with $\omega$. (A3) is satisfied if $\omega_\mathcal{R}$ is mixing (returns to equilibrium) for the C*-dynamical system $(\mathcal{A}_\mathcal{R},\alpha^\mathcal{R})$ \cite{AJPP}. This assumption implies that if the initial state is not far from $\omega$, it converges to the same NESS.

Under the above assumption $\gamma^{-1}\circ(\delta_L-\delta_R)\circ\gamma$ is translationally invariant. So there is a translationally invariant interaction $\Psi$ (not necessarily of finite range) such that
\[ \gamma^{-1}\circ(\delta_L-\delta_R)\circ\gamma(A)=i\left[\sum_{X\cap\Lambda\neq\emptyset}\Psi(X),A\right],\ A\in\mathcal{A}_\Lambda. \]

Now let us define the work density. Let $\Phi_t$ be a cyclic operation (time dependent interaction) such that $\Phi_t=\Phi$ for $t\in(-\infty,0]\cup[T,\infty)$ and $\|\Phi_\cdot\|<\infty$. The state transformation $\theta_T$ induced by the operation is the thermodynamical limit of the state transformation $\theta_T^{\Lambda_N}$ induced by the time dependent Hamiltonian $H_{\Lambda_N}(\Phi_t)$;
\[ \theta_T(A)=\lim_{N\to\infty}\theta_T^{\Lambda_N}(A) ,\ A\in\mathcal{A} \]
(the existence of this limit is due to Lemma 3.3).
$\theta_T^{\Lambda_N}$ can be written as $\theta_T^{\Lambda_N}(A)=U_T^{N*}AU^N_T$ by the unitary element $U_T^N\in\mathcal{A}_{\Lambda_N}$. Define the work density $w(\Phi_\cdot)$ extracted from the system in $\omega_+$ by the cyclic operation $\Phi_t$ by
\[ w(\Phi_\cdot)=\lim_{N\to\infty}\frac{1}{|\Lambda_N|}(\omega_+(H_{\Lambda_N}(\Phi))-\omega_+\circ\theta_T(H_{\Lambda_N}(\Phi))) \]
if the limit exists, that is, the difference of the energy density in the initial state and in the state after the operation. For example this limit exists in the case of translationally invariant operations. It should be remarked that there is a cyclic operation $\Phi_t$ such that $w(\Phi_\cdot)>0$. This fact is proven from the result of \cite{PW} and translation invariance of NESS $\omega_+$, although the proof is not constructive. To estimate this value we introduce a lemma.

\begin{lem}
\[ \lim_{N\to\infty}\frac{1}{|\Lambda_N|}\|\theta_T(H_{\Lambda_N}(\Phi))-\theta_T^{\Lambda_N}(H_{\Lambda_N}(\Phi))\|=0 .\]
\end{lem}
This lemma follows from Lemma 3.4. and the evaluation of the boundary term. By this lemma we have
\begin{eqnarray*}
 w(\Phi_\cdot)&=&\lim_{N\to\infty}\frac{1}{|\Lambda_N|}(\omega_+(H_{\Lambda_N}(\Phi))-\omega_+(U_T^{N*}H_{\Lambda_N}(\Phi)U_T^N))\\
&=&\lim_{N\to\infty}\frac{1}{|\Lambda_N|}i\omega_+(U_T^{N*}\delta_V(U_T^N)) .
\end{eqnarray*}
Set
\[ c(\beta_L,\beta_R)=\sup_{\|\Phi_\cdot\|<\infty}\varliminf_{N\to\infty}\frac{1}{|\Lambda_N|}\omega_+\left(U_T^{N*}\sum_{X\cap\Lambda_N\neq\emptyset}\Psi(X)U_T^N\right)-\omega_+\left(\sum_{X\ni0}\Psi(X)\right). \]
\begin{lem}
Under the assumptions $\mathrm{(A1),(A2),(A3)}$, we have
\[ 0\le c(\beta_L,\beta_R)\le 2\left\|\displaystyle\sum_{X\ni0}\Psi(X)\right\| \]
and for any cyclic (and even in the case of the fermion system) operations $\Phi_t$ with $\|\Phi_\cdot\|<\infty$,
\[ w(\Phi_\cdot)+\frac{2}{\beta_L+\beta_R}s(\Phi_\cdot)\le\frac{\beta_L-\beta_R}{\beta_L+\beta_R}c(\beta_L,\beta_R) \]
holds, where
\[ s(\Phi_\cdot)=\varlimsup_{N\to\infty}\frac{1}{|\Lambda_N|}S(\omega_+\circ\theta_T^{\Lambda_N}\|\omega_+). \]
\end{lem}
\begin{proof}
Obviously $c(\beta_L,\beta_R)\ge0$.
By the inequality of Lemma 3.2. and Lemma 3.5.
\[ w(\Phi_\cdot)+\frac{2}{\beta_L+\beta_R}s(\Phi_\cdot)\le\frac{\beta_L-\beta_R}{\beta_L+\beta_R}\left[\varliminf_{N\to\infty}\frac{1}{|\Lambda_N|}\omega_+\left(U_T^{N*}\sum_{X\cap\Lambda_N\neq\emptyset}\Psi(X)U_T^N\right)-\omega_+\left(\sum_{X\ni0}\Psi(X)\right)\right], \]
and for any $\Phi_t$ with $\|\Phi_\cdot\|<\infty$
\[ w(\Phi_\cdot)+\frac{2}{\beta_L+\beta_R}s(\Phi_\cdot)\le\frac{\beta_L-\beta_R}{\beta_L+\beta_R}c(\beta_L,\beta_R) \]
holds. Next we show the inequality $c(\beta_L,\beta_R)\le2\left\|\displaystyle\sum_{X\ni0}\Psi(X)\right\|$.
Since $\Psi$ is translationally invariant
\begin{eqnarray*}
\left\|\sum_{X\cap\Lambda_N\neq\emptyset}\Psi(X)\right\|&=&\left\|\sum_{X\ni-N}\Psi(X)+\sum_{X\ni-N+1,X\not\ni-N}\Psi(X)+\cdots+\sum_{X\ni-1,X\not\ni-N,\cdots,N-1}\Psi(X) \right\| \\
&\le&\left\|\sum_{X\ni0}\Psi(X)\right\|+\left\|\sum_{X\ni0,X\not\ni-1}\Psi(X)\right\|+\cdots+\left\|\sum_{X\ni0,X\not\ni-2N,\cdots,-1}\Psi(X)\right\| .
\end{eqnarray*}
Let $S_N=\{=N,-N+1,\cdots,-1\}$. Due to the property of the interaction $\Psi$ that $E_Y(\Psi(X))=0$ for $Y\subsetneq X$,
\[ \sum_{X\ni0,X\not\ni-N,\cdots,-1}\Psi(X)=E_{S_N}\left(\sum_{X\ni0,X\not\ni-N,\cdots,-1}\Psi(X)\right)=E_{S_N}\left(\sum_{X\ni0}\Psi(X)\right) .\]
Since $E_{S_N}$ is a conditional expectation, $\|E_{S_N}(A)\|\le\|A\|$ and we obtain 
\[ \left\|\sum_{X\cap\Lambda_N\neq\emptyset}\Psi(X)\right\|\le|\Lambda_N|\left\|\sum_{X\ni0}\Psi(X)\right\|. \]
\[ \frac{1}{|\Lambda_N|}\omega_+\left(U_T^{N*}\sum_{X\cap\Lambda_N\neq\emptyset}\Psi(X)U_T^N\right)\le\frac{1}{|\Lambda_N|}\left\|\sum_{X\cap\Lambda_N\neq\emptyset}\Psi(X)\right\| \]
and this completes the proof.
\end{proof}
$\left\|\displaystyle\sum_{X\ni0}\Psi(X)\right\|$ is finite due to the equation (2). Since the relative entropy is non-negative, we obtain the following inequality giving the upper bound of the work density extracted by the cycle $\Phi_t$.
\begin{thm}[the 2nd law-type work relation in NESS]
Under the assumptions $\mathrm{(A1),(A2),(A3)}$, we have
\begin{equation}
 w(\Phi_\cdot)\le\frac{\beta_L-\beta_R}{\beta_L+\beta_R}c(\beta_L,\beta_R) .
\end{equation}
\end{thm}
This inequality is different from that of Theorem 3.3. at the point that the upper bound is independent of the operations. So this inequality implies that we cannot extract the work density more than this bound from NESS by any cyclic operations. In the equilibrium limit $\beta_L\downarrow\beta_R$ this bound tends to 0 and the inequality reduces to the second law of thermodynamics. 

Finally we want to give a remark on the realizability of the bound in the inequality (2). It can be proved that this inequality is not tight in the sense that one can show that this bound cannot be achieved by the following theorem.      
\begin{thm}
\[ w(\Phi_\cdot)\neq0\Rightarrow s(\Phi_\cdot)>0 .\]
\end{thm}
\begin{proof}
Most of the tools necessary to the proof are in \cite{AM}. Denote $H_{\Lambda_N}(\Phi)$ and $H_{\Lambda_N}(\Phi+\lambda\Psi)$ by $H_N$ and $K_N$, where $\lambda=\frac{\beta_L-\beta_R}{\beta_L+\beta_R}$. Set $W_N=\displaystyle\sum_{X\cap\Lambda_N\neq\emptyset,X\cap\Lambda_N^C\neq\emptyset}(\Phi(X)+\lambda\Psi(X))$. The key of the proof is the following equality on the relative entropy\cite{BR2}: Suppose that the cyclic vector $\Omega_2$ associated to $\omega_2$ is separating for $\pi_2(\mathcal{A})''$ and $\omega_1(A)=\langle\Omega_1,\pi_2(A)\Omega_1\rangle$ for cyclic and separating vector  $\Omega_1$ in $\mathcal{H}_2$, then for $P=P^{*}\in\mathcal{A}$
\[ S(\omega_1\|\omega_2)=S(\omega_1\|\omega_2^P)+\omega_1(P)-\log\|\Omega_2^P\|^2,\]
holds, where  $\omega_2^P(A)=\frac{\langle\Omega_2^P,\pi(A)\Omega_2^P\rangle}{\|\Omega_2^P\|^2}$ and $\Omega_2^P=e^{\frac{\log\Delta+\pi_2(P)}{2}}\Omega_2$ ($(\mathcal{H}_2,\pi_2,\Omega_2)$ is the GNS triple associated to $\omega_2$ and $\Delta$ is the modular operator induced by the cyclic and separating vector $\Omega_2$).

First consider the case $P=\beta W_N, \omega_1=\omega_+\circ\theta_T^{\Lambda_N}$ and $\omega_2=\omega_+$, where $\beta=\frac{\beta_L+\beta_R}{2}$. Then $\omega_+^{\beta W_N}=\omega_\beta^N\otimes\phi$ (Gibbs condition), where $\omega_+^N(A)=\frac{\mathrm{Tr}e^{-\beta K_N}A}{\mathrm{Tr}e^{-\beta K_N}}, A\in\mathcal{A}_{\Lambda_N}$, and we obtain
\[ S(\omega_+\circ\theta_T^{\Lambda_N}\|\omega_+)=S(\omega_+\circ\theta_T^{\Lambda_N}\|\omega_\beta^N\otimes\phi)+\beta\omega_+\circ\theta_T^{\Lambda_N}(W_N)-\log\|\Omega_+^{\beta W_N}\|^2 .\]
$\Omega_+$ is the cyclic vector corresponding to $\omega_+$. Next consider the case $P=\alpha H_N,\omega_1=\omega_+\circ\theta_T^{\Lambda_N}$ and $\omega_2=\omega_\beta^N\otimes\phi$ ($\alpha\in\mathbb{R}$ is arbitrary).
By the positivity of the relative entropy, we have 
\[ S(\omega_+\circ\theta_T^{\Lambda_N}\|\omega_+)\ge\alpha\omega_+\circ\theta_T^{\Lambda_N}(H_N)+\beta\omega_+\circ\theta_T^{\Lambda_N}(W_N)-\log\frac{\mathrm{Tr}e^{-\beta K_N+\alpha H_N}}{\mathrm{Tr}e^{-\beta K_N}}-\log\|\Omega_+^{\beta W_N}\|^2 .\]
Recall that $\frac{1}{|\Lambda_N|}\|W_N\|\to0\ (N\to\infty)$\cite{AM}. By the Peierls-Bogoliubov inequality and the Golden-Thompson inequality 
\[ e^{\beta\omega_+(W_N)}\le\|\Omega_+^{\beta W_N}\|^2\le\|e^{\frac{\beta\pi_+(W_N)}{2}}\Omega_+\|^2\le e^{\beta\|W_N\|} \]
and 
\[ \lim_{N\to\infty}\frac{1}{|\Lambda_N|}\log\|\Omega_+^{\beta W_N}\|^2=0. \]
Set
\[ e_T=\lim_{N\to\infty}\frac{1}{|\Lambda_N|}\omega_+\circ\theta_T^{\Lambda_N}(H_N)\] 
\[ e=\lim_{N\to\infty}\frac{1}{|\Lambda_N|}\omega_+(H_N) .\]
Then we have
\[ s(\Phi_\cdot)\ge\alpha e_T-f(\alpha)+f(0) \]
where $f(\alpha)=\displaystyle\lim_{N\to\infty}\frac{1}{|\Lambda_N|}\log\mathrm{Tr}e^{-\beta K_N+\alpha H_N}$ (this limit exists\cite{AM}). Since $\omega_+$ is a unique KMS state, by the equivalence of the KMS condition and the variational principle,  $f(\alpha)$ is differentiable at $\alpha=0$ and $\frac{d}{d\alpha}f(\alpha)|_{\alpha=0}=e$.  
Here we consider the case $w(\Phi_\cdot)>0$ (The case $w(\Phi_\cdot)<0$ is similarly proved). Fix $\epsilon$ such that $0<\epsilon<w(\Phi_\cdot)$. We can choose $\alpha<0$ such that 
\[ \left|\frac{f(\alpha)-f(0)}{\alpha}-e\right|<\epsilon .\]
Then 
\begin{eqnarray*}
s(\Phi_\cdot)&\ge&\alpha e_T-\alpha(e-\epsilon) \\
&=&\alpha(\epsilon-w(\Phi_\cdot)) \\
&>&0.
\end{eqnarray*}
\end{proof}

\section{Example: free lattice fermion}
In this section we discuss the free Fermi gas on $\mathbb{Z}$. This system satisfies all the assumptions (A1)-(A3) in the previous section. NESS of free fermion systems is already well studied\cite{AJPP,AJPP2}. Let us start with the general free Fermi gas.

Let $\mathcal{A}_\mathcal{S}=\mathcal{A}^{CAR}(\mathcal{H}_\mathcal{S})$,\ $\mathcal{A}_\mathcal{R}=\mathcal{A}^{CAR}(\mathcal{H}_L\oplus\mathcal{H}_R)$, where $\mathcal{H}_\mathcal{S}$ is a finite dimensional Hilbert space. In this case the C*-algebra of  the whole system $\mathcal{A}=\mathcal{A}_\mathcal{S}\otimes\mathcal{A}_\mathcal{R}=\mathcal{A}^{CAR}(\mathcal{H}_\mathcal{S})\otimes\mathcal{A}^{CAR}(\mathcal{H}_L\oplus\mathcal{H}_R)$ is isomorphic to $\mathcal{A}^{CAR}(\mathcal{H}_L\oplus\mathcal{H}_\mathcal{S}\oplus\mathcal{H}_R)$ \cite{AJPP}. So we identify $\mathcal{A}$ with $\mathcal{A}^{CAR}(\mathcal{H}_L\oplus\mathcal{H}_\mathcal{S}\oplus\mathcal{H}_R)$. Let $h_\mathcal{S},h_L,h_R$ be self-adjoint operators on $\mathcal{H}_\mathcal{S},\mathcal{H}_L,\mathcal{H}_R$ (1-particle Hamiltonian) and define
\[ g=\beta_Lh_L\oplus0\oplus\beta_Rh_R \]
\[ h_0=h_L\oplus h_\mathcal{S}\oplus h_R .\]
on $\mathcal{H}_L\oplus\mathcal{H}_\mathcal{S}\oplus\mathcal{H}_R$. The initial state $\omega_\mathcal{S}\otimes\omega_\mathcal{R}$ on $\mathcal{A}_\mathcal{S}\otimes\mathcal{A}_\mathcal{R}$ corresponds to the quasi-free state on $\mathcal{A}=\mathcal{A}^{CAR}(\mathcal{H}_L\oplus\mathcal{H}_\mathcal{S}\oplus\mathcal{H}_R)$ generated by the self-adjoint operator $\frac{1}{1+e^g}$ on $\mathcal{H}=\mathcal{H}_L\oplus\mathcal{H}_\mathcal{S}\oplus\mathcal{H}_R$. $\omega$ on $\mathcal{A}^{CAR}(\mathcal{H})$ is said to be a quasi-free state if there is a self-adjoint operator $T$ on $\mathcal{H}$ such that $0\le T\le I$ and 
\[ \omega(a^{*}(f_n)\cdots a^{*}(f_1)a(g_1)\cdots a(g_m))=\delta_{nm}\mathrm{det}((\langle g_j,Tf_i\rangle)_{ij}), \]
$n,m\in\mathbb{N},\ f_i,g_j\in\mathcal{H}$.
Denote $\alpha_t$ the dynamics induced by $h_0$. For a trace class self-adjoint operator $v=v_L+v_R$ on $\mathcal{H}$, $V=\mathrm{d}\Gamma(v)\in\mathcal{A}$ and the dynamics induced by $h=h_0+v$ and the perturbed dynamics $\alpha^V_t$ coincide. So
\[ \alpha_{-t}\circ\alpha^V_t(a^{\#}(f))=a^{\#}(e^{-ith_0}e^{ith}f). \]
From this relation if the limit s-$\displaystyle\lim_{t\to\infty}e^{-ith_0}e^{ith}(=W)$ exists, then the limit s-$\displaystyle\lim_{t\to\infty}\alpha_{-t}\circ\alpha^V_t(=\gamma)$ also exists. $W$ is called a wave operator and plays an important role in scattering theory. By Kato-Rosenblum\cite{Ka,RO}, if $h$ has only absolutely continuous spectrum, then this limit exists and $\mathrm{Ran}W=\mathcal{H}_{ac}(h_0)$, $\mathrm{Ran}\gamma=\mathcal{A}^{CAR}(\mathcal{H}_{ac}(h_0))\subset\mathcal{A}$. For a self-adjoint operator $A$ on $\mathcal{H}$, $\mathcal{H}_{ac}(A)$ is the set of $\phi\in\mathcal{H}$ such that $\langle\phi,P_A(\cdot)\phi\rangle$ is absolutely continuous with respect to the Lebesgue measure, where $P_A(\cdot)$ is the spectral measure of $A$.  When $h,h_L,h_R$ have only absolutely continuous spectrum, (A1), (A2) are satisfied.

Now consider a concrete model, free Fermi gas on one-dimensional lattice $\mathbb{Z}$. Let $M\in\mathbb{N}$ and assume that Hilbert spaces are given by
\[ \mathcal{H}_L=l^2((-\infty,-M-1]\cap\mathbb{Z}) \]
\[ \mathcal{H}_\mathcal{S}=l^2([-M,M]\cap\mathbb{Z}) \]
\[ \mathcal{H}_R=l^2([M+1,\infty)\cap\mathbb{Z}) ,\]
so $\mathcal{H}_L\oplus\mathcal{H}_\mathcal{S}\oplus\mathcal{H}_R=l^2(\mathbb{Z})$. 1-particle Hamiltonian including the interaction $v$ is 
\[ (h\psi)(n)=-\frac{1}{2}(\psi(n-1)+\psi(n+1)) .\]
This is the discrete version of $-\frac{1}{2}\Delta$ on $L^2(\mathbb{R})$ ($\Delta$ is Laplacian).

The interaction between the system and the reservoirs are
\[ v_L=-\frac{1}{2}(\langle e_{-M},\cdot\rangle e_{-M-1}+\langle e_{-M-1},\cdot\rangle e_{-M}) \]
\[ v_R=-\frac{1}{2}(\langle e_{M},\cdot\rangle e_{M+1}+\langle e_{M+1},\cdot\rangle e_{M}) ,\]
where $\{e_n\}_{n\in\mathbb{Z}}$ is the standard basis of $l^2(\mathbb{Z})$. The total Hamiltonian is formally given by 
\[ -\frac{1}{2}\sum^{\infty}_{n=-\infty}(a^{*}_na_{n+1}+a^{*}_{n+1}a_n) \]
and
\[ V=-\frac{1}{2}(a^{*}_{-M}a_{-M-1}+a^{*}_{-M-1}a_{-M}+a^{*}_Ma_{M+1}+a^{*}_{M+1}a_M)\in\mathcal{A}. \]
On the momentum space $L^2(-\pi,\pi)$, $h$ is the multiplication operator $\hat{h}$,
\[ (\hat{h}\hat{\psi})(k)=(-\cos k)\hat{\psi}(k). \]
This operator has only absolutely continuous spectrum. By Kato-Rosenblum theorem s-$\displaystyle\lim_{t\to\infty}\alpha_{-t}\circ\alpha_t^V$ exists. Since $\mathcal{H}_{ac}(h_0)=\mathcal{H}_L\oplus\mathcal{H}_R$, $\mathrm{Ran}\gamma=I\otimes\mathcal{A}_\mathcal{R}$ holds. Thus, this system satisfies (A1), (A2) in the previous section. Furthermore $\omega_\mathcal{R}$ is mixing for $(\mathcal{A}_\mathcal{R},\alpha^\mathcal{R})$ and faithful. So (A3) is also satisfied and we can apply Theorem 3.7. It is known that NESS $\omega_+$ is the quasi-free state generated by the multiplication operator on $L^2(-\pi,\pi)$ of the function
\[ 
\rho(k)=
\begin{cases}
\frac{1}{1+e^{-\beta_R\cos k}} & (-\pi< k <0)\\
\frac{1}{1+e^{-\beta_L\cos k}} & (0\le k\le\pi) .
\end{cases}
\]
Furthermore this NESS is the KMS state with the dynamics corresponding to the following interaction \cite{MO},
\[ \Phi(X)=
\begin{cases}
-\frac{1}{2}(a^{*}_na_{n+1}+a^{*}_{n+1}a_n) & (X=\{n,n+1\}) \\
-\frac{2i}{\pi}\frac{\beta_L-\beta_R}{\beta_L+\beta_R}\frac{2l}{(2l)^2-1}(a^{*}_na_{n+2l}-a^{*}_{n+2l}a_n) & (X=\{n,n+2l\},l\in\mathbb{Z})\\
0 & (\mathrm{otherwise}).
\end{cases}
\]
Thus in this model the interaction $\Psi$ in the Lemma 3.6. is 
\[ \Psi(X)=
\begin{cases}
-\frac{2i}{\pi}\frac{2l}{(2l)^2-1}(a^{*}_na_{n+2l}-a^{*}_{n+2l}a_n) & (X=\{n,n+2l\})\\
0 & (\mathrm{otherwise})
\end{cases}
\]
and we can estimate 
$\left\|\displaystyle\sum_{X\ni0}\Psi(X)\right\|$ as
\[ \left\|\sum_{X\ni0}\Psi(X)\right\|\le2\sqrt{1-\frac{4}{\pi^2}} .\]

\end{document}